\documentclass[11pt]{article}

\usepackage{amsfonts, color}
\usepackage{epsfig}
\usepackage{latexsym}
\usepackage{amsmath}
\usepackage{amssymb}
\usepackage{amsthm}
\usepackage{float}
\def\bequ{\begin{equation}}
\def\eequ{\end{equation}}
\def\barr{\begin{array}}
\def\earr{\end{array}}

\def\ben{\begin{equation}}
\def\een{\end{equation}}
\def\bena{\begin{eqnarray}}
\def\eena{\end{eqnarray}}

\def\b1{e^0}

\newcommand{\be}{\begin{equation}}
\newcommand{\ee}{\end{equation}}
\def\bea{\begin{eqnarray}}
\def\eea{\end{eqnarray}}
\def\Tr{\mathrm{Tr}}

\newtheorem{theorem}{Theorem}[section]
\newtheorem{lemma}[theorem]{Lemma}
\newtheorem{prop}[theorem]{Proposition}

\newtheorem{defin}[theorem]{Definition}

\def\be{\begin{equation}}
\def\ee{\end{equation}}
\def\bea{\begin{eqnarray}}
\def\eea{\end{eqnarray}}

\def\lesssim{\mathrel{\hbox{\rlap{\hbox{\lower4pt\hbox{$\sim$}}}\hbox{$<$}}}}
\def\gtrsim{\mathrel{\hbox{\rlap{\hbox{\lower4pt\hbox{$\sim$}}}\hbox{$>$}}}}

\begin{document}
\hfuzz=100pt
\title{Non-existence of Skyrmion--Skyrmion and 
Skyrmion--anti-Skyrmion static equilibria}
\author{G. W. Gibbons$^1$, C.M. Warnick$^{1, 3}$ and W. W. Wong$^2$ 
\\
\\ \small{1. D.A.M.T.P., University of Cambridge, Wilberforce Road, Cambridge CB3 0WA,
  U.K.} 
\\ \small{2. D.P.M.M.S., University of Cambridge, Wilberforce Road,
  Cambridge CB3 0WB, U.K.} 
\\ \small{3. Queens' College, Cambridge, CB3 9ET, U.K.}}

\date{May 14, 2010}       
 \maketitle  \let\thefootnote\relax\footnotetext{{\em Email}:
  g.w.gibbons@damtp.cam.ac.uk, c.m.warnick@damtp.cam.ac.uk,
  w.wong@dpmms.cam.ac.uk \\ \mbox{} \hspace{.4cm} {\em DAMTP pre-print no.:} DAMTP-2010-40}

\begin{abstract}
We consider classical static Skyrmion--anti-Skyrmion
and Skyrmion--Skyrmion configurations, symmetric
with respect to a reflection plane,
or symmetric up to a $G$-parity transformation respectively.
We show that the stress tensor component completely normal to the
reflection plane, and hence its integral over the plane,
is negative definite or positive definite respectively.
Classical Skyrmions always repel classical
Skyrmions and classical Skyrmions
always attract classical anti-Skyrmions and thus no
static equilibrium, whether stable or unstable,
is possible in either case. No other symmetry assumption
is made and so our results also apply to multi-Skyrmion configurations.
Our results are consistent
with existing analyses of Skyrmion forces at large separation,
and with numerical results on Skymion--anti-Skyrmion
configurations in the literature which  admit a different
reflection symmetry. They also hold for the massive Skyrme model.
We also point out that reflection symmetric
self-gravitating Skyrmions or black holes with Skyrmion hair
cannot rest in symmetric equilibrium with  self-gravitating
anti-Skyrmions.
\end{abstract}

\section{Introduction}

Most of our present ideas about the interactions of  matter and 
anti-matter are quantum mechanical and  
go back to Dirac's discovery of the equation named after him \cite{Dirac1} and 
his subsequent development of hole theory to interpret
its  negative energy solutions \cite{Dirac2, Dirac3}. However the concept
of matter-antimatter symmetry had been  anticipated in the nineteenth century
\cite{Schuster} and was manifest in the earlier classical  attempts 
to explain  electrical and gravitational forces in terms
of   bodies, pulsating in-phase or out-of-phase   in the 
aether \cite{Forbes, Pearson1, Pearson2},
or  as  sources or sinks  of currents in the aether \cite{Pearson3}    
possibly arriving from, or departing to,  extra dimensions \cite{Rouse}.
The former idea resembles the properties of non-topological solitons,
Q-matter and boson stars \cite{Lee, Friedberg, ColemanQ}. The latter idea  is a curious anticipation of current models
in string theory in which the ends of strings attached to  3-branes
appear  as positively and negatively charged particles 
described by the Born-Infeld action
\cite{BornInfeld,Gibbons1}.

Common to both classical and quantum mechanical
models is the idea that (charged)  particles   repel  particles
while  particles attract anti-particles. The many successes of quantum 
field theory in treating particle--anti-particle interactions
did not mean however that classical models of matter were  abandoned.
Dirac himself introduced  magnetic monopoles and anti-monopoles \cite{Dirac3}
and 't Hooft and Polyakov \cite{Hooft, Polyakov} showed that
Yang-Mills-Higgs theory admits  
smooth classical solutions with finite energy. Earlier Born and Infeld
\cite{BornInfeld}, motivated by the self-energy problem in electrodynamics,
had resuscitated  a particular example of a class of 
non-linear electrodynamical
theories of Mie \cite{Mie} and showed that the theory admits
finite energy, albeit non-smooth particle and anti-particle like solutions
called BIons.
This whole  line of development can be said to have culminated in 
Skyrme's introduction of the model that bears his name \cite{Skyrme1,Skyrme2}
in which smooth classical finite energy solutions, called Skrymions
or anti-Skyrmions,
correspond to baryons and anti-baryons.  For a review of the
entire field of the role of  classical solutions in quantum field theory
including Skyrmions see \cite{MantonSutcliffe}.  
For a very recent account of all aspects of Skyrmions as they
are currently  used in Nuclear Physics see \cite{BrownRho}.

A common feature of all of the classical models described above
is the existence of both  particle and anti-particle like solutions,
generically known as solitons, 
such that soliton number minus anti-soliton number is conserved
but neither soliton number nor anti-soliton number is separately conserved
since solitons and anti-solitons may annihilate. 
In the case of topological solitons, annihilation is possible 
because a soliton--anti-soliton configuration is topologically trivial.
 
Loosely speaking, the notion that particles obey the ``likes repel,
opposites attract'' rule manifests in these classical models as
repulsive and attractive forces between soliton--soliton and
soliton--anti-soliton configurations respectively, and one does not
expect such configurations can remain in stable equilibrium. However 
the concept of force between particles, except in the limit
of large separation, is not well-defined  for solitons, because
the notion of a soliton position is itself not well-defined.
Since we are dealing with a continuum theory it is better
to think of the stress tensor $T_{ij}=T_{ji}$
in terms of which the total force  acting  on a surface $S$
is
\ben F_i = \int _S T_{ij} d s _j \,.\een 
If a soliton, or anti-soliton, is effectively enclosed
within a domain $D$, we take $S$ as the boundary  $ S= \partial D$.
It is difficult to say much about the force
in the general time-dependent case, but  if the system is stationary
we have
\ben
\partial _i T_{ij}=0\,.
\een
It follows by the divergence theorem that for any domain 
$D$ the total force must vanish
\ben
\int _{\partial D} T_{ij} d s_j =0\,.
\een 
A convenient special case is when $D$ is a half space
given by $z \ge  0$,
and  the stress tensor falls off sufficiently rapidly at infinity
so that the integral over the  large hemisphere 
$\sqrt{x^2+y^2 +z^2} = r = \mbox{constant}$, $z\ge 0 $
 vanishes  as 
 $r \rightarrow \infty$. It follows that for an equilibrium to be possible one must have
\ben
\int _{z=0} T_{iz}\, dx dy =0\,.
\een
This criterion was used in section 4.1 of \cite{Gibbons1}
to calculate the force between BIons and between
BIons and anti-BIons, and hence to  rule out the existence of certain 
static BIon-BIon and BIon--anti-BIon configurations.
The same idea was adapted to incorporate gravity and black holes
in \cite{BeigGibbonsSchoen}.
The purpose of the present paper is to extend the discussion to
the case of Skyrmions in flat spacetime.

An important ingredient of the arguments in \cite{Gibbons1,BeigGibbonsSchoen}
is the assumption that the configurations whose existence
one  wishes to exclude admit a reflection symmetry $r_P$ 
which fixes a plane $P$  and which acts on the space of fields
as a suitable particle number reversing involution. 
The use of  a reflection map was  
also used in \cite{Coleman} to show that the energy of
a reflection symmetric  monopole anti-monopole 
configuration may always be reduced by moving them closer together.
The argument was extended to quantum-mechanical Casimir
forces \cite{Kenneth} and a rationale later provided
in terms of the quantum field theoretic concept of Reflection
Positivity \cite{Bachas}.  This strongly suggests that the use
of a matter anti-matter reflection plane is not merely
a technical device, but has a deeper significance. We shall comment further
on this point in the conclusion.

\section{Reflection Symmetry and Lagrangian Field Theory } 

\subsection{Mathematical Preliminaries}

We first describe the mathematical conventions. Throughout we consider
a Lagrangian field theory for maps $\Phi$ from Minkowski space
$(\mathbb{R}^{1,d},\eta)$ into a Riemannian manifold
$(\mathcal{N},h)$ (we state our main mathematical result for arbitrary
spatial dimensions, though in subsequent application we will focus on
the $d=3$ case). For concreteness we take the signature convention $(-+\cdots+)$ for
the Minkowski metric, though the mathematical statements in this paper
are independent of the convention chosen. Let the action be defined by some generally
covariant Lagrangian density $\mathcal{L}$. $\Phi$, being formally a
stationary point of this action, solves an associated Euler-Lagrange
equation. The Einstein-Hilbert stress-energy tensor (hereon we refer
to as the stress tensor) is defined by the variation of the Lagrangian
density relative to the inverse metric
\be\label{eq:stresstensor}
T_{\alpha\beta} = \left . \frac{-2}{\sqrt{|g|}} \frac{\delta
\mathcal{L}}{\delta g^{\alpha\beta}} \right |_{g = \eta}
\ee
where $\alpha,\beta$ are space-time indices. That $\Phi$ solves the
Euler-Lagrange equation is captured in the statement that the stress
tensor is divergence free
\be
\eta^{\alpha\beta}\nabla_\alpha T_{\beta\gamma} = 0\,.
\ee
For an arbitrary vector field $X^\alpha$, we can construct the energy
current ${}^{(X)}J_\beta = T_{\alpha\beta}X^\alpha$. A simple
computation shows that
\be
2\eta^{\alpha\beta}\nabla_\alpha {}^{(X)}J_\beta =
T_{\alpha\beta}{}^{(X)}\pi^{\alpha\beta}
\ee
where the covariant \emph{deformation tensor relative to} $X^\alpha$ is 
 ${}^{(X)}\pi_{\alpha\beta} = (L_X\eta)_{\alpha\beta}$ given by the
Lie-derivative of the metric tensor by the vector field $X^\alpha$. In
particular, we have Noether's theorem, which states that continuous
symmetries of the space-time generate conserved currents. 

We are interested in static solutions, for which $\partial_0\Phi = 0$.
In this case the stress tensor is independent of time, and on a
constant-time slice $\cong (\mathbb{R}^d,e)$ is divergence free
\be
\partial_i T_{ij} = 0
\ee
(Roman indices are spatial only). On the spatial slice, the coordinate
vector fields $\partial_1,\ldots,\partial_d$ generate translation
symmetries. Then by the divergence theorem, 
\be
\int_{\partial D} {}^{(X)}J_i ds_i = 0
\ee
for any translation vector field $X^i$ with the integral evaluated on
the boundary of some domain $D\subset\mathbb{R}^d$. If we pick $D$ to be a
particular half-space of $\mathbb{R}^d$, and $X^i$ the translation
field normal to its boundary (which is a $d-1$ dimensional
hyperplane), we conclude that
\be
\int_P T(n,n) ds = 0
\ee
where $P$ is a hyperplane, $ds$ the induced surface measure, and $n$
a unit normal to $P$, if we assume that its corresponding current
decays sufficiently fast at infinity. (In particular, it suffices that
the stress tensor decays as $|T_{ij}(x)| \lesssim (1 +
|x|)^{(1-d-\epsilon)}$.)

In the following we will also make the assumption that the system is
isolated. The criterion that we will impose is then that
$\Phi$ extends continuously to a map from $\mathbb{S}^d\to
\mathcal{N}$ under one-point compactification of the spatial slice.
For a $C^1$ solution, a sufficient condition to guarantee the
criterion is that $|\partial\Phi| \lesssim (1+|x|)^{-1-\epsilon}$.

In the case of the harmonic map system or the Skyrme system, the two
decay assumptions mentioned above are less restrictive than finite
total energy.

\subsection{Even and odd reflections}

Given a hyperplane $P\subset\mathbb{R}^d$, there is a natural action
of $\mathbb{Z}_2$ by reflection which we will denote by $r_P$. In this
section we consider solution maps to the Lagrangian field theory outlined
above that are equivariant under reflection. We shall assume that
there is a natural $\mathbb{Z}_2$ action on $\mathcal{N}$ \emph{whose
fixed points are discrete}, and we denote the action by the negative
sign $-$. In the case where $(\mathcal{N},h)$ is a symmetric space, we can
fix some point $q\in \mathcal{N}$ and let $-$ be the geodesic symmetry
map about the point $q$. In the case where $(\mathcal{N},h)$ is a Lie
group with a bi-invariant metric, we take $-$ to be the group
inverse. For the case of the Skyrme model, which we shall consider
below, $-$ corresponds to what is known to physicists as $G$-parity \cite{Michel}.

\begin{defin}
Fixing a hyperplane $P$, a solution $\Phi$ to the static
Euler-Lagrange equation is said to be \emph{even}
across $P$ if $\Phi(r_Px) = \Phi(x)$, and \emph{odd} if $\Phi(r_Px) = -\Phi(x)$. 
\end{defin}
Under the decay assumptions assumed, a continuous odd solution $\Phi$
will take a value at some fixed-point of the $\mathbb{Z}_2$ action on
$\mathcal{N}$,  $q = -q$,  along the hyperplane $P$, 
and hence take the same value $q$ at infinity. 

The content of this letter is in the following convexity conditions
\begin{defin}
Fix a point $x$ and a unit tangent vector $W$ at $x$. Let $Q$ be the
orthogonal complement of $\{W\}$ in the tangent space at $x$. A 
(static) Lagrangian field theory
is said to be \emph{even-semidefinite} with respect to $Q$ 
if there exists $s_e\in \{-1,+1\}$ such that for any solution $\Phi$
satisfying $\nabla_W\Phi(x) = 0$, the corresponding stress tensor
obeys $s_eT(W,W)|_x \geq 0$. The theory is said to be
\emph{odd-semidefinite} with respect to $Q$ if there exists 
$s_o\in \{-1,+1\}$ such that for 
any solution satisfying $\nabla_Y\Phi(x) = 0~ \forall Y\in Q$, the
stress tensor obeys $s_oT(W,W)|_x \geq 0$. We replace ``semidefinite''
by ``definite'' if equality is satisfied only when $\nabla\Phi(x)
\equiv 0$.  
\end{defin}
The natural symmetries of Euclidean space guarantees that if a theory
is even/odd-(semi)definite with respect to some $Q$, it will be so
with respect to any $Q$. This condition is intimately tied to the dominant energy condition for
the associated dynamical theory. In the case where $s_es_o = -1$, a
Wick rotation argument suggests that the dominant energy
condition follows from even- and odd-definiteness if the contribution 
from the cross terms to the energy is negligible (by cross terms we 
mean terms of the
form $\nabla_W\Phi\cdot\nabla_Y\Phi$ where $Y\in Q$).

A simple example of a theory that is both even- and odd-definite is 
the harmonic map system. The stress tensor is given by
\[ T^{\mathrm{harm}}_{ij} = \nabla_i\Phi^A\nabla_j\Phi_A - \frac12
e_{ij} \nabla_k\Phi^A\nabla^k\Phi_A \]
so taking an orthonormal frame $\{f_{(i)}\}$ at the point $x$, we have
\[ T^{\mathrm{harm}}_{(1)(1)} = \frac12
\nabla_{(1)}\Phi^A\nabla_{(1)}\Phi_A - \frac12 \sum_{2\leq k \leq d}
\nabla_{(k)}\Phi^A\nabla_{(k)}\Phi_A \]
and the condition is satisfied with $s_e = -1$ and $s_o = +1$. As we
shall see, the static Skyrme system is also even- and odd-definite. 

The main lemma is the following
\begin{lemma}\label{mainlemma}
Let $P$ be a hyperplane and $\Phi$ be a classical even(odd) solution to an 
even(odd)-definite static field theory that satisfies the decay
conditions given above, and that $\Phi$ converges to some point $q\in
\mathcal{N}$ at infinity. Then $\Phi|_P = q$ and $\nabla\Phi |_P = 0$. 
\end{lemma}
\begin{proof}
Observe that if $\Phi$ is an even solution, then its normal derivative
to $P$ is 0. If $\Phi$ is an odd solution, than its derivatives
tangential to $P$ vanish. So if the field theory is even/odd-definite,
$T(n,n)$ is signed for $n$ the unit normal vector to $P$, unless
$\nabla\Phi \equiv 0$. By our decay
condition we guaranteed that $\int_P T(n,n) ds = 0$, so $T(n,n)$ must
vanish point-wise along $P$, and thus all first derivatives of $\Phi$ 
(normal and tangential to $P$) vanish there. Integrating back from 
infinity we have $\Phi|_P = q$. 
\end{proof}

\subsection{Non-existence of even/odd solutions of Skyrme system}

That the dynamical Skyrme model satisfies the dominant energy condition has been
established in \cite{Gibbons2} (see also \cite{Wong}); an analogous
computation also shows that the static Skyrme system is both even- and
odd-definite, with constants $s_e = -1$ and $s_o = 1$. We start by
reviewing the static Skyrme model. 

In this model, the source manifold is $\mathbb{R}^3$ with standard
Euclidean metric $e$, and the target is $\mathcal{N} = SU(2)$ with the
bi-invariant metric associated to the Killing form. Recall that
$SU(2)$ is the Lie group of unitary $2\times 2$ complex matrices with
unit determinant, its Lie algebra $\mathfrak{su}(2)$ is the collection
of $2\times 2$ traceless anti-Hermitian matrices. The Killing form
$\mathfrak{su}(2)\times\mathfrak{su}(2)\to \mathbb{R}$ is negative
definite and given by $\Tr(VW)$ where $V,W\in\mathfrak{su}(2)$
multiply by standard matrix multiplication. The standard
bi-invariant metric $h$ is customarily normalized to be equal to
negative one-half of the Killing form at identity, such that $i$ times
the Pauli matrices form an orthonormal basis. 

Since the metric is bi-invariant, it is convenient to express the
Lagrangian density in terms of the $\mathfrak{su}(2)$ valued current
$L_i = U^\dagger \partial_iU$. The energy/action functional $E$ is
then
\be
E = - \int d^3 x \sqrt{|e|} \left( \frac{1}{2} \Tr\left(L_i L^i\right) +
\frac{1}{16} \Tr
\left([L_i, L_j][L^i, L^j] \right)\right) \label{energy}
\ee
where $L^i = e^{ij}L_j$. The Euler-Lagrange equation for the energy
minimizer can be written in divergence form as
\be\label{eq:ELE}
\nabla_i\left(L^i - \frac14 [L_j,[L^j,L^i]]\right)=0.
\ee
The energy and hence the field equations are manifestly invariant
under $SU(2)\times SU(2)$ acting on the Skyrme field $U$ by
\be\label{eq:SU2action}
U(x) \to A_1U(x)A_2\,, \qquad A_1,A_2\in SU(2)\,.
\ee
Therefore, without loss of generality, we can assume that the limit at
infinity of our Skyrme field is the Identity element, and break the
$SU(2)\times SU(2)$ symmetry to just $SU(2)$, acting as in
\eqref{eq:SU2action} but with $A_1 = A_2^\dagger$. 

This boundary
condition implies a one-point compactification of $\mathbb{R}^3$ to
$S^3$ so that topologically $U:S^3 \to S^3$. Thus there is a
topological charge given by the degree of this map, known as the
baryon number, which may be expressed as an integral over the baryon
number density given by the pullback of the volume form on $SU(2)$ to $\mathbb{R}^3$:
\begin{equation}
B = -\frac{1}{24 \pi^2} \int d^3 x \sqrt{|e|} \epsilon^{ijk}
\Tr(L_iL_jL_j)\,.
\end{equation}
Under a spatial reflection, the baryon number density changes sign since
the orientation of the spatial manifold is reversed. Under the
combination of a spatial reflection and a $G$-parity transformation
the baryon number density is preserved, since there is an
additional change of orientation on the target manifold. Thus an even
field configuration with positive baryon number on one side of the
plane of symmetry will have negative baryon number on the other side
(loosely a Skyrmion--anti-Skyrmion configuration). Conversely an
odd field configuration with positive baryon number on one side of the
plane of symmetry will also have positive baryon number on the other
side (loosely a Skyrmion-Skyrmion configuration). For more general
models whose target space is not $SU(2)$ this interpretation will
depend on the nature of the $\mathbb{Z}_2$ action on the target space and the
topological charge.

We can directly compute the stress tensor via \eqref{eq:stresstensor}.
With the notation $L_{ij} = -\Tr(L_iL_j)$, 
\be\label{eq:Skyrmestress}
T_{ij} = L_{ij} - \frac12 L_k{}^kg_{ij} - \frac12\left( L_{ik}L_j{}^k
- \frac14 g_{ij}L_{kl}L^{kl} - L_{ij}L_k{}^k + \frac14
  g_{ij}(L_k{}^k)^2\right).
\ee
As $L_{ij}$ is the pull-back of twice the bi-invariant metric on
$SU(2)$ we may pick a basis at $x$, orthonormal with respect to $e_{ij}$,
such that $L_{ij}(x) =
\mathrm{diag}(\lambda_1,\lambda_2,\lambda_3)$, $\lambda_k$ being
non-negative real numbers. In this basis
$T_{ij}(x)$ is also diagonal, with 
\be
T_{(1)(1)}(x) = \frac12 (\lambda_1 - \lambda_2 - \lambda_3) +
\frac14(\lambda_1\lambda_2 + \lambda_1\lambda_3 -
\lambda_2\lambda_3)
\ee
and $T_{(2)(2)}(x), T_{(3)(3)}(x)$ given by cycling indices. 

\begin{prop}\label{skyrmeisdefinite}
The static Skyrme model is both even- and odd-definite.
\end{prop}
\begin{proof}
Observe that if for some non-vanishing vector $X$, we have that
$\nabla_XU(x) = 0$, then $X^iL_i(x) = 0$ and hence $X^iL_{ij}(x) = 0$. This
implies that such an $X$ is an eigenvector of $L_{ij}(x)$. Now let $W$ be
a unit vector and $Q$ its orthogonal complement. Suppose $\nabla_WU(x) =
0$. Then we can complete $f_{(1)} = W$ to an orthonormal basis such that
$\lambda_1 = 0$. Then 
\[ T(W,W) = T_{(1)(1)} = -\frac12(\lambda_2+\lambda_3) -
\frac14(\lambda_2\lambda_3) \]
is manifestly non-positive, with $0$ attained only with $\lambda_2 =
\lambda_3 = 0$ (which implies that $\nabla U(x) = 0$). Similarly, if
$\nabla_YU(x) = 0$ for any $Y\in Q$, then $Q$ is in the null-space of
$L_{ij}$, and hence we can complete $f_{(1)} = W$ to an orthonormal
basis such that $\lambda_2 = \lambda_3 = 0$, and 
\[ T(W,W) = T_{(1)(1)} = \frac12 \lambda_1 \]
is non-negative, with $0$ attained only when $\nabla U(x) = 0$. 
\end{proof}

We note that the addition of a mass term of the form
\be
E_m = \int  m^2 \Tr (1-U) \sqrt{|e|} d^3 x
\ee
to the energy (\ref{energy}) does not affect the results of the above
proposition. In the case of an even reflection it contributes to the
stress tensor with the correct sign, while for an odd reflection it
has no contribution to the stress tensor. Physically one expects this
to be the case since the mass term affects the range of the forces,
but not their signs.

\begin{theorem}
A smooth solution to the static Skyrme model with sufficient decay
at infinity that is either even or odd about a hyperplane $P$ must be
the vacuum solution $U \equiv Id$. 
\end{theorem}
\begin{proof}
Combining Proposition \ref{skyrmeisdefinite} and Lemma \ref{mainlemma}
we have that any solution satisfying the hypotheses of the theorem
must be constant along $P$ with vanishing normal derivative. By
inspection of the Euler-Lagrange equations \eqref{eq:ELE}, all
derivatives of any order of $U$ must vanish along $P$. By Taylor's
theorem with remainder, $U$ must converge to the identity element
faster than any polynomial in a neighborhood of $P$. Then by the
strong unique continuation property of the static Skyrme model (see
Theorems 2.4 and 2.5 in \cite{MantonSchroersSinger}) $U$ must be
constant. 
\end{proof}

We remark that the solution $\Phi$ vanishes to infinite order along
$P$ under the conclusion of Lemma \ref{mainlemma} is generically true
for field theories whose Euler-Lagrange equation is second order
elliptic, in view of the fact that the conclusion of the lemma forces
simultaneously a Dirichlet and a Neumann boundary condition. An 
analogous theorem can be then proven for any such system
provided that it has the unique continuation property. 

\subsection{Well-separated Skyrmions}

As is typically the case with topological solitons, the field
configurations believed to minimize the energy (\ref{energy}) with a
given charge $B$ have a compact `core' region in which the energy is
concentrated and an asymptotic region where the fields are close to
the vacuum. This is made precise in Theorem 2.1 of
\cite{MantonSchroersSinger} where it is shown that under fairly weak
assumptions on $U$, any solution of the Skyrme equations
(\ref{eq:ELE}) has a complete asymptotic expansion in powers of $1/r$
and $\log r$. In the case that $U$ is non-constant, the leading term
of this expansion is harmonic, hence a multipole. 

To discuss the asymptotics it is convenient to introduce the triplet
of pion fields $\pi_a$ in terms of which $U$ is written as
\be
U = \sigma + i \pi_a \tau_a
\ee
where $\tau_a$ are the Pauli matrices and $\sigma$ depends on the pion
fields through the constraint $\sigma^2 + \pi_a \pi_a = 1$. We will
focus our attention on the $B=1$ soliton, believed to be the
`hedgehog' configuration found by minimizing the energy within a
spherically symmetric class, see \cite{MantonSutcliffe}. For
this configuration, the asymptotic field is known to have the leading
order form
\be
\pi_a = \frac{C}{|x|^2} x_a\,, \qquad \sigma = 1\,.
\ee
This corresponds to a triplet of orthogonal pion dipoles. It can
additionally be shown that if a second $B=1$ Skyrmion is introduced in
the far field of the first, the interactions between them are at
leading order dominated by the dipole-dipole interactions in the pion
fields \cite{Manton}.

\begin{figure}[!ht]
\[
\begin{picture}(0,0)%
\includegraphics{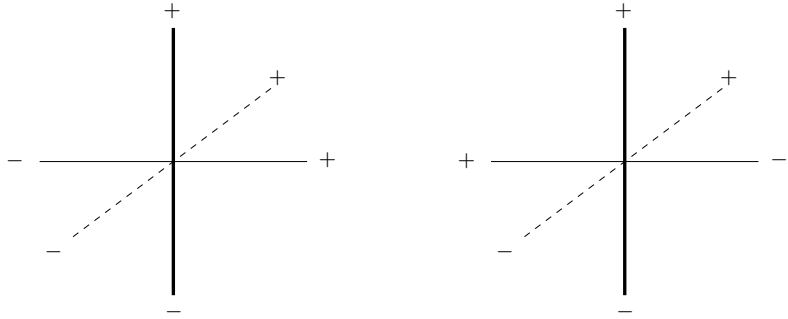}%
\end{picture}%
\setlength{\unitlength}{2763sp}%
\begingroup\makeatletter\ifx\SetFigFont\undefined%
\gdef\SetFigFont#1#2#3#4#5{%
  \reset@font\fontsize{#1}{#2pt}%
  \fontfamily{#3}\fontseries{#4}\fontshape{#5}%
  \selectfont}%
\fi\endgroup%
\begin{picture}(7101,2848)(2101,-5396)
\put(2446,-4816){\makebox(0,0)[lb]{\smash{{\SetFigFont{8}{9.6}{\rmdefault}{\mddefault}{\updefault}{\color[rgb]{0,0,0}$-$}%
}}}}
\put(4906,-3991){\makebox(0,0)[lb]{\smash{{\SetFigFont{8}{9.6}{\rmdefault}{\mddefault}{\updefault}{\color[rgb]{0,0,0}$+$}%
}}}}
\put(2101,-4006){\makebox(0,0)[lb]{\smash{{\SetFigFont{8}{9.6}{\rmdefault}{\mddefault}{\updefault}{\color[rgb]{0,0,0}$-$}%
}}}}
\put(4456,-3256){\makebox(0,0)[lb]{\smash{{\SetFigFont{8}{9.6}{\rmdefault}{\mddefault}{\updefault}{\color[rgb]{0,0,0}$+$}%
}}}}
\put(3511,-2656){\makebox(0,0)[lb]{\smash{{\SetFigFont{8}{9.6}{\rmdefault}{\mddefault}{\updefault}{\color[rgb]{0,0,0}$+$}%
}}}}
\put(3526,-5356){\makebox(0,0)[lb]{\smash{{\SetFigFont{8}{9.6}{\rmdefault}{\mddefault}{\updefault}{\color[rgb]{0,0,0}$-$}%
}}}}
\put(6496,-4816){\makebox(0,0)[lb]{\smash{{\SetFigFont{8}{9.6}{\rmdefault}{\mddefault}{\updefault}{\color[rgb]{0,0,0}$-$}%
}}}}
\put(6151,-4006){\makebox(0,0)[lb]{\smash{{\SetFigFont{8}{9.6}{\rmdefault}{\mddefault}{\updefault}{\color[rgb]{0,0,0}$+$}%
}}}}
\put(7561,-2656){\makebox(0,0)[lb]{\smash{{\SetFigFont{8}{9.6}{\rmdefault}{\mddefault}{\updefault}{\color[rgb]{0,0,0}$+$}%
}}}}
\put(8506,-3256){\makebox(0,0)[lb]{\smash{{\SetFigFont{8}{9.6}{\rmdefault}{\mddefault}{\updefault}{\color[rgb]{0,0,0}$+$}%
}}}}
\put(8956,-3991){\makebox(0,0)[lb]{\smash{{\SetFigFont{8}{9.6}{\rmdefault}{\mddefault}{\updefault}{\color[rgb]{0,0,0}$-$}%
}}}}
\put(7576,-5356){\makebox(0,0)[lb]{\smash{{\SetFigFont{8}{9.6}{\rmdefault}{\mddefault}{\updefault}{\color[rgb]{0,0,0}$-$}%
}}}}
\end{picture}%
\]
\caption{Even reflection -- Skyrmion--anti-Skyrmion pair}
\end{figure}
\begin{figure}[!ht]
\[
\begin{picture}(0,0)%
\includegraphics{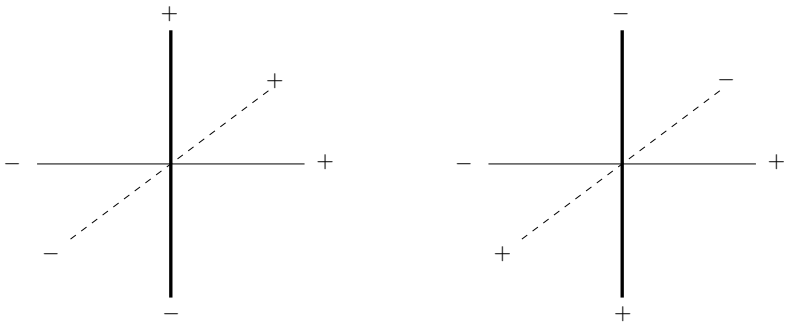}%
\end{picture}%
\setlength{\unitlength}{2763sp}%
\begingroup\makeatletter\ifx\SetFigFont\undefined%
\gdef\SetFigFont#1#2#3#4#5{%
  \reset@font\fontsize{#1}{#2pt}%
  \fontfamily{#3}\fontseries{#4}\fontshape{#5}%
  \selectfont}%
\fi\endgroup%
\begin{picture}(7147,2848)(2101,-5396)
\put(2446,-4816){\makebox(0,0)[lb]{\smash{{\SetFigFont{8}{9.6}{\rmdefault}{\mddefault}{\updefault}{\color[rgb]{0,0,0}$-$}%
}}}}
\put(4906,-3991){\makebox(0,0)[lb]{\smash{{\SetFigFont{8}{9.6}{\rmdefault}{\mddefault}{\updefault}{\color[rgb]{0,0,0}$+$}%
}}}}
\put(2101,-4006){\makebox(0,0)[lb]{\smash{{\SetFigFont{8}{9.6}{\rmdefault}{\mddefault}{\updefault}{\color[rgb]{0,0,0}$-$}%
}}}}
\put(4456,-3256){\makebox(0,0)[lb]{\smash{{\SetFigFont{8}{9.6}{\rmdefault}{\mddefault}{\updefault}{\color[rgb]{0,0,0}$+$}%
}}}}
\put(3511,-2656){\makebox(0,0)[lb]{\smash{{\SetFigFont{8}{9.6}{\rmdefault}{\mddefault}{\updefault}{\color[rgb]{0,0,0}$+$}%
}}}}
\put(3526,-5356){\makebox(0,0)[lb]{\smash{{\SetFigFont{8}{9.6}{\rmdefault}{\mddefault}{\updefault}{\color[rgb]{0,0,0}$-$}%
}}}}
\put(7576,-5356){\makebox(0,0)[lb]{\smash{{\SetFigFont{8}{9.6}{\rmdefault}{\mddefault}{\updefault}{\color[rgb]{0,0,0}$+$}%
}}}}
\put(8956,-3991){\makebox(0,0)[lb]{\smash{{\SetFigFont{8}{9.6}{\rmdefault}{\mddefault}{\updefault}{\color[rgb]{0,0,0}$+$}%
}}}}
\put(6496,-4816){\makebox(0,0)[lb]{\smash{{\SetFigFont{8}{9.6}{\rmdefault}{\mddefault}{\updefault}{\color[rgb]{0,0,0}$+$}%
}}}}
\put(6151,-4006){\makebox(0,0)[lb]{\smash{{\SetFigFont{8}{9.6}{\rmdefault}{\mddefault}{\updefault}{\color[rgb]{0,0,0}$-$}%
}}}}
\put(8506,-3256){\makebox(0,0)[lb]{\smash{{\SetFigFont{8}{9.6}{\rmdefault}{\mddefault}{\updefault}{\color[rgb]{0,0,0}$-$}%
}}}}
\put(7561,-2656){\makebox(0,0)[lb]{\smash{{\SetFigFont{8}{9.6}{\rmdefault}{\mddefault}{\updefault}{\color[rgb]{0,0,0}$-$}%
}}}}
\end{picture}%
\]
\caption{Odd reflection -- Skyrmion--Skyrmion pair}
\end{figure}

Let us apply this asymptotic analysis to the reflection symmetric
situation we consider above. There are two possibilities, shown in
Figures 1 and 2. We represent the dipoles in each of the three pion
fields with different styles of lines and take the reflection plane vertically. Since each reflection, whether
in real space or field space, reverses the baryon number we see that
the even reflection corresponds to a Skyrmion--anti-Skyrmion pair while
the odd reflection corresponds to a Skyrmion--Skyrmion pair.

Recalling that for scalar
fields, like the pions, like charges attract, we see that the
Skyrmion--anti-Skyrmion pair feel a mutual attraction while the
Skyrmion--Skyrmion pair feel a mutual repulsion. Thus neither may
remain in a static equilibrium. This argument will also to apply to
configurations with higher baryon number which may be thought of, in
the asymptotic field, as a composite of $B=1$ Skyrmions. This
heuristic argument provides a physical justification for the rigorous
results of the previous section which apply even when the field
configurations may not be approximated as consisting of two well
separated solitons.

\subsection{Other Symmetry Types}
In the preceding we have considered the two simplest cases of
$\mathbb{Z}_2$ action on the target manifold $\mathcal{N}$. There are,
of course, other candidates for an ``equivariant'' solution. In
general, we can classify the symmetries by the dimension of its
set of fixed points. 

The even symmetry we have discussed corresponds to the set of fixed
points having the same dimension as the target manifold. The identity
map on $\mathcal{N}$ is the only isometry with such property. The odd
symmetry, on the other hand, corresponds to the set of fixed points
being zero dimensional. For symmetric spaces as target manifolds,
reflections about totally geodesic submanifolds are examples of
$\mathbb{Z}_2$ actions with fixed point sets of intermediate numbers of
dimensions. 

In the case of the Skyrme model, the target manifold is $SU(2)$, which
has 3 real dimensions. We can consider actions of $\mathbb{Z}_2$ that
fixes a 1-dimensional subset (i.e.~a $U(1)$ subgroup) or a
2-dimensional subset (i.e.~an equatorial 2-sphere through the
identity) compatible with our decay assumption at infinity. Recall
that elements of $SU(2)$ can be presented as pairs
$(a,b)\in\mathbb{C}^2$ such that $a\bar{a} + b\bar{b} = 1$. Examples
of an action that fixes $U(1)$ include $(a,b) \to (\bar{a},\bar{b})$ 
and $(a,b) \to (a,-b)$, all of which are in fact
conjugate in $SU(2)$. An example of an action that fixes a 2-sphere is
$(a,b) \to (\bar{a},b)$. In any of these cases, since the plane of
symmetry $P\subset \mathbb{R}^3$ is allowed to be mapped into a
non-trivial subset of $\mathcal{N}$, the transverse stress $T(n,n)$ is
no longer guaranteed to have a sign, and the arguments given in Lemma
\ref{mainlemma} no longer apply. In fact, as was shown
numerically by Krusch and Sutcliffe \cite{KruschSutcliffe}, an ansatz
corresponding to $(a,b) \to (a,-b)$ can lead to nontrivial 
Sphaleron solutions in the Skyrme model. In fact, Sphaleron solutions
with chains of Skyrmion--anti-Skyrmion pairs have been shown to exist \cite{Shnir:2009ct}.

 \subsection{Self-Gravitating Skyrmions} 
   
Since the pioneering work of Luckock  and Moss \cite{Luckock}
there have been many studies in which the Skyrme model
is coupled to general relativity. Static solutions
may or may not contain  black holes. Since gravity is always attractive
one can say little in general about whether Skyrmion-Skyrmion
equilibria  are possible. If they are, they are presumably unstable.
However one can say, using the methods of
\cite{BeigGibbonsSchoen} and the results of this paper
that  Skyrmion--anti-Skyrmion equilibria, with a reflection symmetry
whether stable or unstable
are not possible.

\section{Conclusion}

In this paper we have  considered classical static Skyrmion--anti-Skyrmion
and Skyrmion-Skyrmion configurations which  are symmetric
with respect to plane reflection,
or symmetric  up to a $G$-parity transformation respectively.
\bea
U(x,y,z) &=& U(x,y,-z) \,, \label{anti}\\
U(x,y,z) &=& U^{-1} (x,y,-z) \,. \label{skyrme}
\eea
By calculating the total force acting
on the plane $P = \{z=0\}$
\ben
F = \int _{P} T_{zz} \,dx dy
\een
we have shown that $F$ is always an attractive force
in the first case and a repulsive force in the second  case.
It is striking that we did not need to make any further symmetry
assumptions and so our results are rather general
and cover the case of multi-Skyrmion configurations.
Our reflection assumption (\ref{anti}) differs from
that used for Sphaleron configurations \cite{KruschSutcliffe}
which acts as
\ben
U(x,y,z)= e^{i \frac{\scriptstyle \pi}{2} \tau _3 }
U(x,y,-z) e^{-i \frac{\scriptstyle \pi}{2} \tau _3 } \,,
\een
and is consistent  with an analysis of
Skrymion forces at large separation \cite{Manton}.
We have noted that it  is straightforward
to incorporate the effects of gravity along the lines of
\cite{BeigGibbonsSchoen}.

Our result is not restricted to the standard Skyrme model,
higher powers of derivatives  could be added to the action and the
essential  reflection even-ness or reflection odd-ness
property of the stress tensor would still hold,
as would the dominant energy condition \cite{Wong}.
One may certainly add a mass term.
The basic ideas should also extend to other groups and other target
spaces,
for example the physically relevant case of  $SU(3)$.

As we remarked briefly in the introduction,
the fact that solitons attract anti-solitons
but repel solitons seems to be very general and we
expect similar results may hold for Yang-Mills monopoles
and possibly for Q-balls, although the latter, being
time-dependent, will certainly bring in new features.

Finally we remark that there is a tantalizing
analogy between our setup, involving as it does
a matter--anti-matter reflection plane  and the interaction energy
of a source and its mirror image, and the setup used
in Euclidean Quantum Field Theory when considering
Reflection Positivity \cite{Jaffe:2007kd}, a connection made originally in \cite{Bachas}.
We have nothing much to say further, other than to remark
that our static  soliton configurations of a theory in
$3+1$ dimensions may of course, by analytic continuation be
regarded as instantons
of a theory in $2+1$ dimensions. While our results
allow us to make a statements about the sign of such forces, the
missing
link is to establish a precise connection between the interaction
energies of instantons and the quantum mechanical inner product
on the  states in Hilbert space associated with the instantons.
Such a connection might be especially illuminating
in view of the fact that, as remarked above,  the extension to include
gravity is immediate and there are already some cases where
reflection positivity has been applied to Euclidean Quantum Gravity
\cite{GibbonsPohle, Jaffe:2006uz, Jaffe:2007uy}.

\section{Acknowledgements}
We would like to thank Costas Bachas and Nick Manton for helpful
comments. W.W.Wong is supported by the Commission of the European Communities, ERC
Grant Agreement No 208007.

\end{document}